\documentclass[10pt, conference, letterpaper]{IEEEtran}
\IEEEoverridecommandlockouts

\usepackage{balance}

\usepackage[utf8]{inputenc}
\usepackage{amsmath,amssymb,amsfonts,amsthm}
\usepackage{graphicx}
\usepackage{textcomp}
\usepackage{xcolor}
\usepackage{array}
\usepackage{booktabs}
\usepackage{pifont}
\usepackage{subfigure}
\usepackage{url}
\usepackage{enumitem}
\usepackage{soul}

\newtheorem{theorem}{Theorem}

\newtheorem{definition}{Definition}

\newtheorem{corollary}{Corollary}

\usepackage{algorithm}
\usepackage{algorithmicx}
\usepackage{algpseudocode}

\usepackage{comment}
\usepackage{enumitem}
\usepackage[normalem]{ulem}
\usepackage{tcolorbox}

\definecolor{purple}{rgb}{1,0,1}
\definecolor{red}{rgb}{1,0,0}
\definecolor{blue}{rgb}{0,0,1}

\newcommand{\ie}{\textit{i.e., }}
\newcommand{\myitem}[1]{\vspace*{0.01in}\noindent\textbf{#1}}
\newcommand{\myitemit}[1]{\vspace*{0.04in}\noindent\textbf{\textit{#1}}}
\newcommand{\myitemitnospace}[1]{\noindent\textbf{\textit{#1}}}

\newcommand{\mysubsubsection}[1]{\myitem{#1}.}

\newcommand{\secref}[1]{\S\ref{#1}}
\newcommand{\extratext}[1]{#1}

\newcommand{\eq}[1]{Eq.~(\ref{#1})}

\usepackage{xspace}

\usepackage{multirow}

\newcommand{\mytableskip}{-\baselineskip}
\newcommand{\myfigureskip}{-\baselineskip}

\title{Estimating the Impact of BGP Prefix Hijacking%
\thanks{This research is co-financed by Greece and European Union through the Operational Program Competitiveness, Entrepreneurship and Innovation under the call RESEARCH-CREATE-INNOVATE (project T2EDK-04937), and the European High-Performance Computing Joint Undertaking (GA No. 951732).}%
}

\author{
Pavlos Sermpezis$^{\mathparagraph}$, Vasileios Kotronis$^{\dagger}$, Konstantinos Arakadakis$^{\dagger, \ddagger}$, Athena Vakali$^{\mathparagraph}$
\\~
\textit{
$^{\mathparagraph}$Aristotle University of Thessaloniki, Greece; $^{\dagger}$FORTH-ICS, Greece;
$^{\ddagger}$University of Crete, Greece
}
}

\begin{document}
\maketitle

\begin{abstract}
BGP prefix hijacking is a critical threat to the resilience and security of communications in the Internet. While several mechanisms have been proposed to prevent, detect or mitigate hijacking events, it has not been studied how to accurately quantify the impact of an ongoing hijack. When detecting a hijack, existing methods do not estimate how many networks in the Internet are affected (before and/or after its mitigation). In this paper, we study fundamental and practical aspects of the problem of estimating the impact of an ongoing hijack through network measurements. We derive analytical results for the involved trade-offs and limits, and investigate the performance of different measurement approaches (control/data-plane measurements) and use of public measurement infrastructure. Our findings provide useful insights for the design of accurate hijack impact estimation methodologies. Based on these insights, we design (i) a lightweight and practical estimation methodology that employs ping measurements, and (ii) an estimator that employs public infrastructure measurements and eliminates correlations between them to improve the accuracy. We validate the proposed methodologies and findings against results from hijacking experiments we conduct in the real Internet. 

\end{abstract}


\maketitle

\IEEEpubidadjcol



\section{Introduction}\label{sec:intro}
The Border Gateway Protocol (BGP) is used by the Autonomous Systems (ASes) to 
establish routing paths in the Internet.
Due to its distributed nature and lack of authentication in the exchanged information, BGP is susceptible to illegitimate route advertisements. \textit{BGP prefix hijacking} is the most prominent example. Numerous hijacking events with global impact on the availability and confidentiality of communications, e.g.,~\cite{google-japan-leak,nanog-squatting,
hijack-BitCoins,hijack-ChinaTelecom
}, and concerns expressed by network operators~\cite{sermpezis2018survey}, show that BGP prefix hijacking is a common and persistent threat to the Internet ecosystem. 

Defenses against BGP prefix hijacking consist of (i) prevention measures, such as prefix filtering or RPKI, which block the propagation of illegitimate routes~\cite{rpki-rfc,bgpsec-specification-2015,
Jumpstarting-BGP-sigcomm-2016
}, and (ii) detection techniques~\cite{sermpezis2018artemis,
Zhang-Ispy-CCR-2008,
heap-jsac2016,
Shi-Argus-IMC-2012
} that inform operators to proceed to counteractions or trigger mitigation techniques~\cite{sermpezis2018artemis,
qiu2010towerdefense}. The efficiency that these mechanisms are expected to have (on average) under various hijacking scenarios has been studied in literature. 
However, \textit{when a hijacking event takes place, there do not exist techniques that can provide accurate information about its actual impact}. 

Knowing the impact of a hijack is important for several reasons: (i) Inform operators about the effect of an ongoing hijack (e.g., global, limited to a few ASes)
. This information may drive their actions, for example, to select mild (e.g., ask other operators to filter hijacked routes) or more aggressive (e.g., prefix de-aggregation) countermeasures to mitigate the hijack, based on its impact~\cite{sermpezis2018artemis}. (ii) Evaluate the actual efficiency of a mitigation action (which is deemed as a key priority by operators~\cite{sermpezis2018survey}). Knowing the remaining impact of the hijack after the mitigation, can help operating and business decisions; for instance, to decide whether further actions are required, or to assess the cost/value of a paid service (e.g., blackholing~\cite{nawrocki2019down} or MOAS announcements
~\cite{sermpezis2018artemis,qiu2010towerdefense}). 


These reasons highlight the need for designing hijack impact estimation techniques, which could be incorporated in existing defense systems and become a valuable asset for network operations and security.
In this direction,
\textit{the goal of this work is to study 
the problem of estimating the impact of an ongoing hijack through measurements, and take the first steps towards designing accurate hijack impact estimation methodologies}. 


Specifically,
we follow an 
approach comprising analysis, simulations, and real experiments and measurements: (i) We analytically study fundamental aspects of the hijack impact estimation, and derive results that identify limits and quantify trade-offs on the accuracy of different estimation methods. (ii) We employ realistic simulations to create datasets of hijack incidents%
\extratext{\footnote{To the best of our knowledge, there are no datasets with real data from past hijacking events that could enable an extensive and in depth investigation.}}%
, based on which we investigate the performance that can be achieved by using different measurement techniques (control/data plane) and available resources (RouteViews, RIPE RIS, RIPE Atlas). (iii) We conduct controlled hijacking experiments 
and extensive network measurements in the real Internet to validate our theoretical/simulation findings.


The main contributions of this work are summarized as:

\begin{itemize}[leftmargin=*]
\item \textbf{\textit{Understanding of the hijack impact estimation.}}~We study the accuracy of the hijack impact estimation under different types of measurements (\secref{sec:understanding}). We show that very high accuracy is possible by sampling/measuring any network in the Internet (e.g., $\sim$1\% estimation error with 1000 samples), while using public measurement infrastructure (RIPE, RouteViews) results in an estimation error of around 10\%. The root cause of this error is the correlation between the locations of public infrastructure monitors.


\item \textbf{\textit{Design of efficient estimators.}} Motivated by our findings, we propose efficient estimation methodologies with and without using public measurement infrastructure. 
We first propose an estimator based on ping campaigns (\secref{sec:ping-estimator}), which does not rely on public infrastructure and can be implemented by any network
. We find that by pinging a couple of reachable IP addresses in a few hundreds of ASes is enough for achieving very low errors. 
Then, we design an estimator based on public infrastructure  measurements (\secref{sec:ridge-estimator}), which employs statistical learning to eliminate the effect of correlation in measurements, and achieves an accuracy comparable to the (best performing) ping-based estimator.
\end{itemize}
To facilitate future research we make our code and data from our experiments available in~\cite{sermpezis2021estimating-code}.

\section{Preliminaries}\label{sec:preliminaries}
We first define the main quantities of the considered problem (\secref{sec:definitions}), provide an overview of 
the different hijack types and their impact characteristics 
(\secref{sec:taxonomy}) and the available network measurement techniques and services (\secref{sec:preliminary-measurements}), and present the simulation (\secref{sec:simulation-methodology}) and experimental (\secref{sec:real-world-exp}) methodology used in the paper.

\subsection{Definitions: Quantities and Metrics}\label{sec:definitions}


\myitemitnospace{Infected AS:} an AS is infected when it routes its traffic to or through the hijacker's AS.

\myitemitnospace{Hijack impact, $I$}: fraction of ASes that are \textit{infected}, $I\in[0,1]$.

\myitemitnospace{Monitor}: an AS for which it can be known (e.g., after a measurement) if it is \textit{infected} or not.

\myitemitnospace{Impact estimator:} a methodology that uses a set of \textit{monitors} and provides an estimation of the \textit{hijack impact}.

Let $I$ be the actual impact of a hijack, and an estimator $\mathcal{E}$ whose estimation is $\hat{I}_{\mathcal{E}}$. The main metrics to characterize the performance of an estimator is the \textit{bias} and the \textit{Root Mean Square Error (RMSE)}: 
\vspace{-0.5\baselineskip}
\begin{equation*}
Bias_{\mathcal{E}} = E[\hat{I}_{\mathcal{E}}-I] 
\hspace{1cm}
RMSE_{\mathcal{E}} = \sqrt{E[(\hat{I}_{\mathcal{E}} - I)^2]}
\end{equation*}
The bias quantifies how far the expected value of the estimator is from the actual value, and the RMSE quantifies the accuracy of the estimator. 
The desired characteristics for an estimator are to be \textit{unbiased} (zero bias) and have low RMSE.

\subsection{
Hijack Types and Impact Characteristics}\label{sec:taxonomy}
We consider an AS that owns and announces a prefix $IP_{*}$; we denote this AS as $AS_{V}$ and call it ``victim AS''. Let the best path to $IP_{*}$ for an $AS_{X}$ be 
\[[AS_{X}, AS_{Y}, ..., AS_{V} | IP_{*}]\] 
where $AS_{X}$ has learned this path through its neighbor $AS_{Y}$. 
%
A hijack takes place when another AS announces an illegitimate path for the prefix $IP_{*}$ or for a more specific prefix. This is the ``hijacker AS'' and we denote it as $AS_{H}$. The hijacker's announcement may propagate to the Internet and ``infect'' some ASes; the extent of the infection is the hijack impact.

There are different ways to perform a hijack and their impact may vary. Below, we present a taxonomy of hijacks~\cite{sermpezis2018artemis, miller2019taxonomy} and discuss their impact characteristics.

\myitem{Origin-AS (Type-0) or Fake-path (Type-N, $N\geq1$) hijack.} In the \textit{origin-AS} hijack the $AS_{H}$ originates the $IP_{*}$ as its own, while in the \textit{fake-path} case the $AS_{H}$ announces a fake path to the $IP_{*}$ to its neighbors; e.g., for an $AS_{X}$

\begin{tabular}{cl}
Type-0: &  $[AS_{X}, AS_{Y}, ..., \mathbf{AS_{H}} | IP_{*}]$\\
Type-N: &  $[AS_{X}, AS_{Y}, ..., \mathbf{AS_{H}}, AS_{Z},..., AS_{V} | IP_{*}]$\\
\end{tabular}\\
where the link $AS_{H}-AS_{Z}$ is fake; the number $N$ denotes that the hijacker's ASN appears in the $N^{th}$ hop away from the origin AS.





\textit{Impact characteristics:} 
In Type-N cases the hijacker originates a longer path than in Type-0 (i.e., with N extra hops). Thus, the paths to the hijacker are longer and less preferred by some ASes; this results in a lower impact for higher $N$~\cite{sermpezis2018artemis,Jumpstarting-BGP-sigcomm-2016}. Note that this holds when no proactive measures, such as RPKI, are deployed; for prefixes protected by RPKI (less than 20\% today~\cite{rpki-monitor-2019,chung2019rpki}), the impact decreases in Type-0 attacks due to route origin validation~\cite{reuter2018towards}, while the impact of Type-N hijacks is not affected since RPKI cannot detect fake links.

\myitem{Exact prefix or Sub-prefix hijack.} The hijacker can perform a Type-0 or Type-N hijack for the same prefix $IP_{*}$ announced by the victim (\textit{exact prefix}) or for a more specific prefix in $IP_{*}$ (\textit{sub-prefix}); for example, let the $IP_{*}$ be the prefix $10.0.0.0/8$, then a sub-prefix hijack takes place if the hijacker announces the prefix $10.0.0.0/9$ (or any $10.0.0.0/d$ with $d\geq9$).

\textit{Impact characteristics:} Default routing in BGP prefers paths to more specific prefixes~\cite{ciscobestpath}. Hence, the impact of a sub-prefix hijack will be larger than an exact prefix hijack; in fact, a sub-prefix hijack will infect the entire Internet (i.e., impact 100\%) 
unless a proactive or 
filtering mechanism is applied. 

\myitem{Data-plane traffic manipulation.} 
For all the aforementioned hijack types, the hijacker can manipulate the traffic that it attracts by: (i) dropping it (\textit{blackholing}), (ii) impersonating a 
service (\textit{imposture}), or (iii) manipulating or eavesdropping~it and then forwarding it to the victim (\textit{man in the middle, MitM}).

\textit{Impact characteristics:} While the traffic manipulation in the data plane by the hijacker does not affect the impact on the control plane (i.e., as defined in this paper), it determines what hijack detection and impact estimation approaches
can be applied (data/control plane measurements) as we discuss~later. 

\vspace{-0.5\baselineskip}
\subsection{Measuring the Hijack Infection}\label{sec:preliminary-measurements}

The hijack impact is determined by the number of infected ASes. Hence, the basic step for an impact estimation is to identify whether an AS/monitor is infected or not. In principle, this can be done by applying any hijack detection method~\cite{sermpezis2018artemis,
Zhang-Ispy-CCR-2008,
heap-jsac2016,
Shi-Argus-IMC-2012
} per~monitor. 

Detection methods are mainly based on three network measurements types. 
Below, we provide some indicative examples for each type, and discuss their main characteristics
. Our first goal in this paper is to investigate how efficient is to use each of these measurement types for impact estimation (see \secref{sec:understanding}). 

\myitem{Route collectors (RC) - BGP routes:} 
A monitor that provides information about its BGP routes (BGP updates or RIBs) can be detected as infected or not (from the AS-path or the prefix in its selected BGP route (see~\cite{sermpezis2018artemis} for a comprehensive approach). For example, for a Type-0 hijack , if the first ASN in the path is different than $AS_{V}$, then the monitor is infected.

The RIPE RIS~\cite{riperis} and RouteViews~\cite{routeviews} projects provide BGP RIBs/updates collected from hundreds ASes
. We refer to these ASes that peer with RIPE RIS / RouteViews route collectors and provide BGP feeds as ``route collector monitors'', or for brevity ``RC''. In the paper, we consider a set of 228 RC that consistently provided data in our experiments (see \secref{sec:real-world-exp}). 

\extratext{The main characteristics of this approach is that it is based on control-plane information, it is lightweight (requires passive measurements, which can be retrieved from public APIs~\cite{ripeatlas}), and can be real-time since several RC provide live-feed of their BGP updates~\cite{routeviews,ripe-ris-real-time,caidabgpstream}.}

\myitem{RIPE Atlas probes (RA) - traceroutes:} Conducting a traceroute from a monitor to the hijacked prefix, returns a path of IP addresses
. Mapping the IPs to ASNs, we can infer the AS-path, and thus detect (similarly to the BGP routes) if the monitor is infected. However, in practice the IP to ASN mapping may be inaccurate for some hops, and advanced methods may be needed to avoid path misinformation~\cite{accurate-traceroute-sigcomm,bdrmapit}.

The RIPE Atlas~\cite{ripeatlas} platform comprises more than $25k$ probes
, \ie devices able to run 
traceroutes towards certain Internet destinations. We refer to the set of ASes with at least one RIPE Atlas probe as ``RA'' monitors, which in our experiments account for 3420 ASes. 

\extratext{This approach combines data plane (traceroute) and control plane (IP-to-ASN mapping) information, and requires active measurements (RIPE Atlas can return a batch of measurements within a few minutes). 
}

\myitem{Pings:} The victim can ping (from its network) an IP address in a remote AS (monitor); if the ping response returns to the victim's network, then the AS can be inferred as not infected (see, e.g., the techniques of~\cite{de2017broad, Zhang-Ispy-CCR-2008}). This inference can be correct in blackholing and imposture hijacks, but not in MitM.

It is important to note that while the first two measurement approaches can use only the monitors of the public services (RC and RA), in this latter case any AS with a responsive pingable IP address (i.e., almost every AS) can be a monitor. 

\extratext{Finally, this approach is based on data-plane information, and requires active measurements (whose results can typically be returned within a few seconds)}

\subsection{Datasets and Simulation Methodology}\label{sec:simulation-methodology}
To study different impact estimation approaches, we would need ground truth data about hijack events and their impact. 
However, typically this information is 
not publicly reported, and detailed datasets do not exist, to the best of our knowledge
. Hence, we use realistic simulations to generate datasets of different hijack types.

Specifically, we simulate the Internet routing system using a largely adopted methodology~\cite{
Jumpstarting-BGP-sigcomm-2016,
sermpezis2018artemis,sermpezis2019pomacs}: (i) we use the AS-relationship dataset~\cite{AS-relationships-dataset} that contains AS-links and inferred inter-AS economic relationships (\textit{customer to provider}, \textit{peer to peer}), based on which (ii) we build the Internet topology graph representing each AS as a single node (a reasonable assumption for the vast majority of ASes~\cite{muhlbauer2006building}) and (iii) we define the routing policies 
as in the Gao-Rexford model~\cite{stable-internet-routing-TON-2001}, where an AS prefers routes learned from its customers, then its peers, and then its providers
, and (iv) we simulate BGP 
using the simulator of~\cite{sermpezis2019pomacs}. For each hijack type, we run $1000$ scenarios with different \{victim, hijacker\} (
\{V,H\}) pairs. Each RC and RA monitor is represented by the AS that hosts it.



While, admittedly, simulations 
may not generate the exact impact output per \{V,H\} case, they have been shown to capture well the routing decisions for the majority of ASes~\cite{anwar2015investigating,sermpezis2019pomacs}. In this work, we study the statistical characteristics of impact estimation rather than the per case behavior, and thus the involved uncertainty is not expected to significantly affect our findings. Nevertheless, we also conduct real hijacking experiments in the Internet (\secref{sec:real-world-exp}), to validate our methods and findings with real data.

\subsection{Real-world Experiments}\label{sec:real-world-exp}
We conduct hijacking experiments in the real Internet using the PEERING testbed~\cite{schlinker2019peering
}.
PEERING owns ASNs and IP prefixes, and has BGP connections with 
networks in several locations (\emph{sites}) around the world. The experiments consist of the following steps
:

\myitem{Selection of \{V,H\} pair.} We create two virtual ASes, assign to them the ASNs $61574$ and $61575$, and connect them to two distinct sites of the PEERING testbed. We select one of them to be the victim AS (V) and the other the hijacker AS (H).


\myitem{BGP announcements and Hijacking.} We conduct Type-0 hijacks, i.e., we announce the prefix $184.164.243.0/24$ from V, and then announce (i.e., hijack) the same prefix from H.

\myitem{Impact measurement: pings (ground truth).} To measure the impact of the hijack, we perform a ping campaign: We select 46k ASes and ping (from a host within PEERING) a reachable IP address in each of them  (see~\secref{sec:ping-estimator})
. We monitor through which PEERING site the ping reply returns: if it returns through the H (or, V) site, we consider the corresponding AS as infected (or, not infected). We use this as the ground truth for the hijack impact of each experiment.

\myitem{Impact measurement: public services.} To apply the different measurement approaches 
(\secref{sec:preliminary-measurements}), we conduct data-plane (traceroutes) and control-plane (BGP updates) measurements after the hijacking announcement: (i) We employ traceroutes from RIPE Atlas probes towards the announced prefix. 
We check in the traceroute the last IP address before it enters PEERING. We map this IP address to an AS (using the prefix-to-AS dataset of~\secref{sec:ping-estimator}), and if the ASN belongs to an upstream provider of the H (or, V) site, then we infer that the AS of the RIPE Atlas probe is infected (or, not infected). 
(ii) Using CAIDA's BGPStream tool~\cite{caidabgpstream} we collect BGP updates received by RouteViews and RIPE RIS monitors. From the AS paths in the BGP updates, we extract the origin ASNs and use them to infer to which site the monitor AS routes its traffic. For example, from the AS path $[AS_X, AS_Y, \dots, AS_H]$, we infer that the monitor $AS_X$ is infected.

In total, we considered a set of 6 PEERING sites that: 
were responsive at the time of our experiments, their BGP announcements propagated to the entire Internet
, and they were reachable through data-plane measurements (pings, traceroutes) from the majority of ASes.
We considered all possible combinations of pairs \{V,H\} for these sites. Omitting the experiments in which the hijack impact was trivial ($100\%$ or $0\%$) or very small/large ($>97\%$ or $<3\%$), we end up in a set of $22$ ``valid'' experiments with different \{V,H\} pairs
. 





\section{\mbox{Understanding the Impact Estimation}}\label{sec:understanding}
In this section, we aim to understand the problem of hijack impact estimation through measurements, and provide useful insights for the design of practical estimation methodologies. 

\subsection{Naive Impact Estimation (NIE)}
The most intuitive approach to estimate the impact of a hijack is to measure a set of monitors, and estimate it as the fraction of infected monitors
. We refer to this approach as the \textit{Naive Impact Estimation}
.

\begin{definition}[Naive Impact Estimator (NIE)]\label{def:naive-impact-estimator}
Let a set of monitors $\mathcal{M}$ ($|\mathcal{M}|=M$), and an indicator $m_{i}$ denoting whether monitor $i\in\mathcal{M}$ is infected ($m_{i}=1$) or not ($m_{i}=0$). The Naive Impact Estimator estimates the hijack impact as 
\begin{equation}\label{eq:definition-nie}
\hat{I}_{NIE(\mathcal{M})} = \textstyle\frac{1}{M}\textstyle\sum_{i\in\mathcal{M}}m_{i}
\end{equation}
\end{definition}

The NIE can be used with any type of measurements (BGP routes, traceroutes, pings) that can provide information to calculate the indicator $m_{i}$. In the following we study the properties and accuracy of NIE, under different types of measurements and monitor sets.

\subsection{Accuracy of the NIE}

\mysubsubsection{1) NIE with Random Set of Monitors}

In the following theorem, we prove that, when the set of monitors $\mathcal{M}$ is randomly selected, the NIE is an \textit{unbiased estimator}, and we derive an expression for its RMSE that is a function of the number of monitors $M$ and the hijack type
.

\begin{theorem}\label{thm:NIE-bias-rmse}
Under a randomly selected set of monitors $\mathcal{M}$, the bias and root mean square error of NIE are given by
\begin{equation*}
Bias_{NIE} = 0 
\text{ \hspace{1cm} }
RMSE_{NIE} = \textstyle\frac{1}{\sqrt{M}}\cdot c_{I}
\end{equation*}
where $c_{I}= \int_{0}^{1} \sqrt{I\cdot (1-I)}\cdot f(I)\cdot dI$, is a constant that depends on the impact distribution $f(I)$.
\end{theorem}
\begin{proof}
The 
proof is given in Appendix~\ref{appendix:proof-NIE-bias-rmse}
. 
\end{proof}

\textit{Remark:} The impact distribution $f(I)$ depends on the \{V,H\} pairs that are expected to be involved in a hijack, and the hijack type.
For example, if any pair of ASes is equally probable to be the \{V,H\} pair, then the impact $I$ of a Type-0 hijack is approximately uniformly distributed in $[0,1]$. 

Table~\ref{tab:hijack-constant} gives the values of the constant $c_{I}$ 
for random \{V,H\} pairs. We also consider scenarios that are closer to reported hijacking activity: namely, scenarios where (i) the \{V,H\} ASes correspond to the events identified as potential hijacks by the BGPmon service in 2018~\cite{
arakadakis2018conext-poster}, and (ii) hijackers are from the set of $22$ ASes classified as ``serial hijackers''~\cite{testart2019profiling} and victims are selected randomly.

\begin{table}[h]
\centering
\caption{Experimentally calculated $c_{I}$ (in parentheses, the average impact $E[I]$) for different hijack types and \{V,H\} pairs.}
\label{tab:hijack-constant}
\vspace{\mytableskip}
\begin{small}
\begin{tabular}{l|ccc}
   & Type-0 & Type-1 & Type-2  \\
\hline
random \{V,H\} pairs         & 0.39 (0.50) & 0.36 (0.30) & 0.31 (0.19)\\
BGPmon \{V,H\} pairs~\cite{arakadakis2018conext-poster}         & 0.35 (0.43) & 0.29 (0.26) & 0.22 (0.17)\\
random V, ``serial'' H~~\cite{testart2019profiling}  & 0.37 (0.68) & 0.40 (0.49) & 0.36 (0.31)
\end{tabular}
\end{small}
\end{table}

\textit{Remark:} The RMSE($I$) is not equal for all values of the real impact $I$ (see Appendix~\ref{appendix:proof-NIE-bias-rmse}). In fact, it is a concave function with a maximum $\frac{1}{2\cdot\sqrt{M}}$ at $I=0.5$ , and minimum $0$ at the corner cases of $I=0$ or $1$. In other words, it becomes more difficult to estimate with high accuracy when the victim and hijacker attract similar fractions of AS routes. This explains the lower values of $c_{I}$ for higher-$N$ hijack types: as $N$ increases, the mass of the impact distribution concentrates around smaller values --closer to $I=0$-- in the area in which the RMSE($I$) takes low values.

Figure~\ref{fig:RMSE-NIE-random-RC-RA-vs-nb-monitors-Type0} shows the RMSE of NIE with a random set of monitors (dashed line) for hijacks of Type-0, as calculated from Theorem~\ref{thm:NIE-bias-rmse} (the simulation results for random monitors almost coincide with the theoretical, i.e., the dashed line, and thus are omitted). The curves for other hijack types and/or \{V,H\} pairs have the same shape, since only the multiplicative factor $c_{I}$ changes (see Table~\ref{tab:hijack-constant}).
\begin{figure}
\centering
\subfigure[theory \& simulations]{\includegraphics[width=0.49\linewidth]{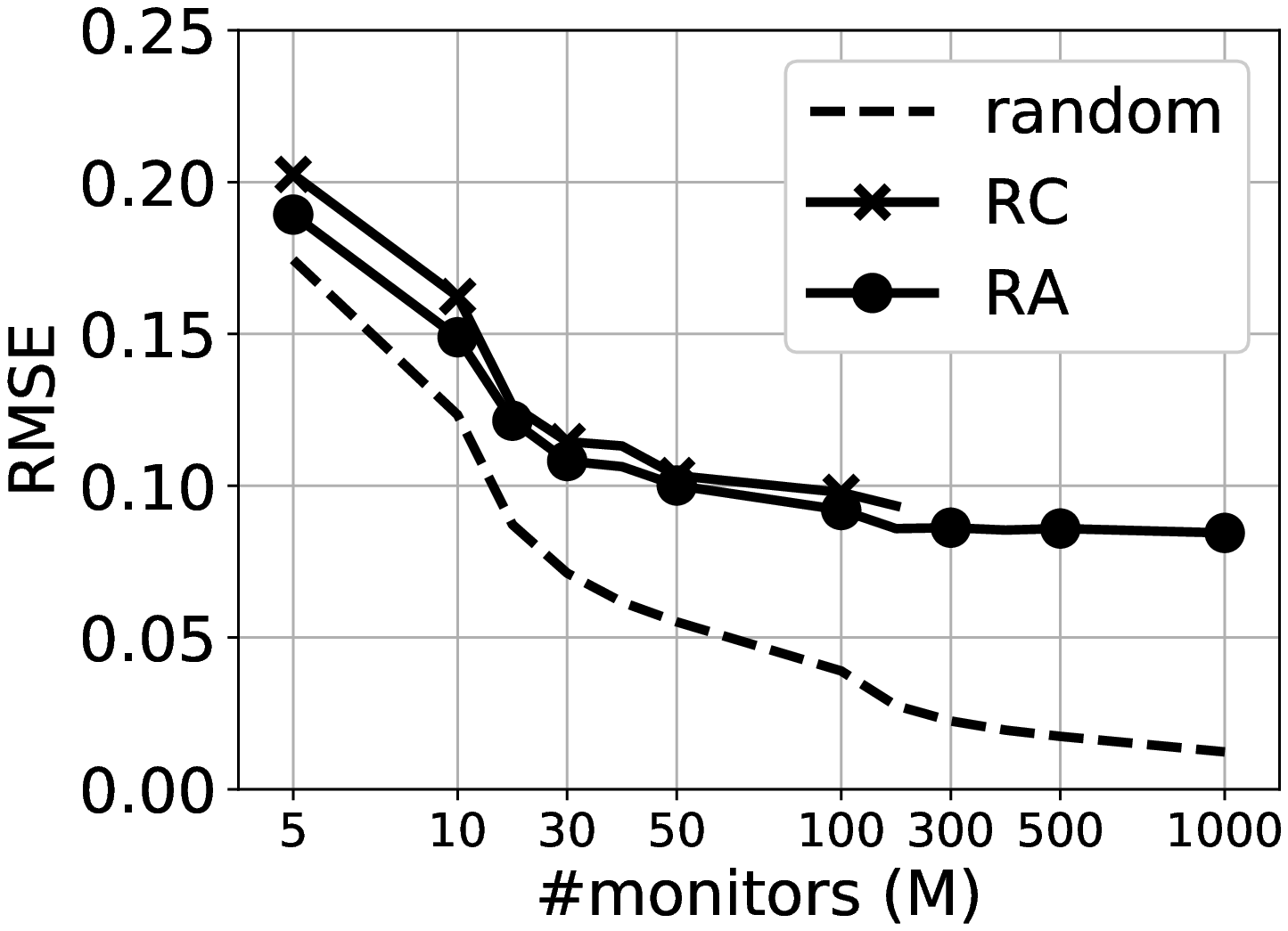}\label{fig:RMSE-NIE-random-RC-RA-vs-nb-monitors-Type0}}
\subfigure[PEERING experiments]{\includegraphics[width=0.49\linewidth]{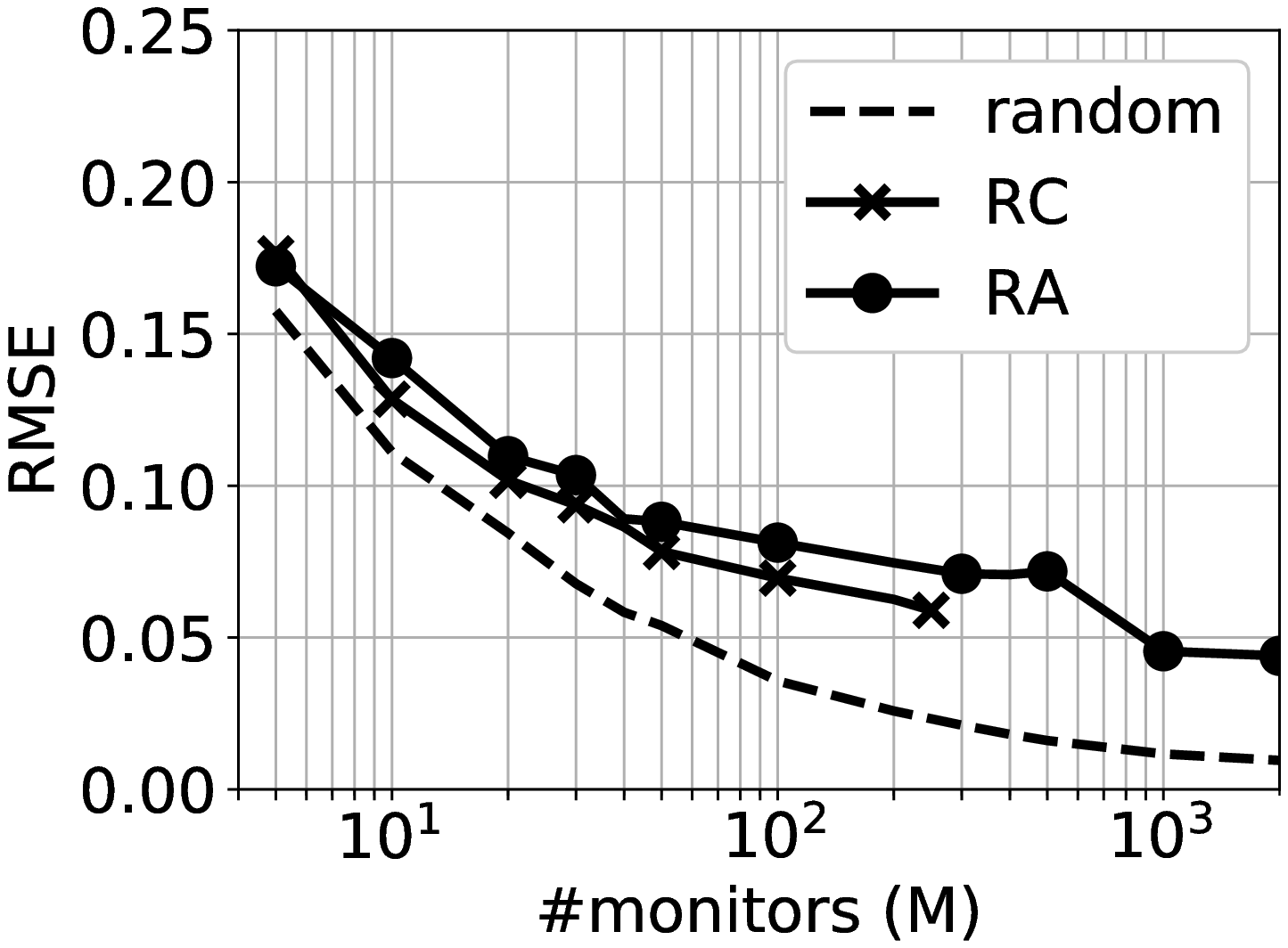}\label{fig:RMSE-NIE-random-RC-RA-vs-nb-monitors-Type0-PEERING}}
\vspace{\myfigureskip}
\caption{RMSE of the NIE (y-axis) vs. number of monitors (x-axis) for Type-0 hijacks and \textit{random}, \textit{RC}, and \textit{RA} sets of monitors. (\textit{Note the different x-axes})}
\label{fig:RMSE-NIE-vs-nb-monitors}
\end{figure}

\mysubsubsection{2) NIE with Monitors from Public Measurement Services}

We now investigate the accuracy of NIE when using the public infrastructure monitors RC and RA (\secref{sec:preliminary-measurements}), which are \textit{not} uniformly located in the Internet~\cite{riperis,routeviews}

\myitemit{Key finding:} \textit{The (non-uniform) locations of the public infrastructure monitors heavily affect the accuracy of NIE.}

Table~\ref{tab:NIE-with-public-monitors}, where we compare simulation results for the RMSE of the NIE using public monitors vs. random sets of (equal number of) monitors, clearly demonstrates that the accuracy heavily depends on the set of employed monitors.



\begin{table}[h]
\caption{RMSE for NIE with public monitors (and random set of monitors) for different hijacks types; simulation results.}
\label{tab:NIE-with-public-monitors}
\centering
\vspace{\mytableskip}
\begin{small}
\begin{tabular}{l|ccc}
{}   &  {Type-0}     &   {Type-1}    & {Type-2}\\
\hline
{RC, 228 monitors}        &  {10\% (2.6\%)}   &   {~9\% (2.4\%)}  & {~8\% (2.1\%)}\\
{RA, 3420 monitors}        &  {~9\% (0.7\%)}   &   {~8\% (0.6\%)}  & {~7\% (0.5\%)}
\end{tabular}
\end{small}
\end{table}

A NIE using the monitors of the public services has a RMSE of 9\%-10\% for Type-0 hijacks, while the same estimator with randomly selected monitors would achieve almost 5 times (2.6\%) and 10 times ($<1\%$) lower RMSE (for the same number of monitors); similar results hold for all hijack types. It is interesting to observe that despite the fact that there are an order of magnitude more data-plane monitors (RA) than control-plane monitors (RC), the accuracy is very similar: RA achieves only 1\% lower RMSE than RC.%
\footnote{
While studying the geographical distribution of RC/RA monitors is out of our scope, Appendix~\ref{appendix:monitors-location} gives some results on its effect on the NIE accuracy.
}


\myitemit{Key finding:} \textit{Measuring 50 monitors of public services, is typically enough for achieving close to the(ir) highest accuracy
.}

Figure~\ref{fig:RMSE-NIE-random-RC-RA-vs-nb-monitors-Type0} compares the RMSE of NIE with random, RC, and RA monitors for Type-0 hijacks as a function of the number of monitors (in the RC and RA cases, we select a random subset of size M in each simulation). We observe that the RMSE of NIE with RC or RA reaches the plateau of around 10\% after 30-50 monitors; for the same $M$, the RMSE of NIE with random monitors is two times lower (around 5\%) and further decreases with the number of monitors. Similar findings hold for the case of hijack Types-1 and 2 as well. 

The experimental results (Fig.~\ref{fig:RMSE-NIE-random-RC-RA-vs-nb-monitors-Type0-PEERING}) are in line with the simulations: (i) public monitors perform consistently worse than random monitors, (ii) the RC and RA curves are similar, and (iii) $M=50$ monitors already achieve 7-8\% RMSE. Note that we use the experiments only for a qualitative validation; the limited number of possible experiments, cannot provide strong statistical significance for the actual RMSE values (e.g., confidence intervals for $M=100$ are $\pm2.7\%$ and $\pm3.4\%$ for RC and RA, respectively).




\subsection{Designing Impact Estimation Methods}
Below we discuss some practical aspects on the implementation of an impact estimator, which --in combination with the above findings-- drive our design for the hijack impact estimation methodologies 
in~\secref{sec:ping-estimator} and~\secref{sec:ridge-estimator}.



\myitem{Random monitors vs. Public infrastructure.} Our results show that selecting monitors randomly (among all ASes in the Internet) results in significantly lower error
. Thus, random monitors are preferable when accuracy is the main goal. However, this approach can be implemented only with ping measurements, since there are no public monitors in all ASes.

\myitem{Ping campaigns: challenges and limitations.} Measuring with pings whether a monitor is infected has some challenges in practice. Pinging an IP address does not necessarily mean that it will reply; in fact, a very small fraction of the addresses in the IP space respond to pings~\cite{pingable-ip-2013,heidemann2008census}. While there are lists of ``pingable'' IP addresses per AS~\cite{antlabhitlists}, they still not always respond to pings. If a monitor does not reply to a ping, we may falsely infer the monitor as infected, and thus overestimate the impact of the hijack. To overcome this challenge, in~\secref{sec:ping-estimator} we first study and quantify the effect of ping failures, and then design a methodology that can still be accurate, by carefully selecting the set and number of IPs per AS to ping.

Finally, a limitation of the ping measurements approach is that it is not applicable to MitM hijacks: all replies will end up to the victim, thus falsely denoting a monitor as non-infected. This can be only overcome with control-plane approaches.


\myitem{Potential of public infrastructure estimations.} 
Using control-plane information (e.g., BGP updates from RC monitors) applies to any hijack type~\cite{sermpezis2018artemis}. Moreover, it can be real-time~\cite{routeviews,ripe-ris-real-time,caidabgpstream}, and implemented by a third-party (i.e., not necessarily the victim network). In this context, and since applying the basic NIE with RC monitors leads to lower accuracy, in \secref{sec:ridge-estimator} we design an estimator more sophisticated than NIE, which uses public infrastructure monitors and achieves comparable performance to the ping-based estimations.


\section{Impact Estimation with Pings
}\label{sec:ping-estimator}

We propose a hijack impact estimation methodology based on ping campaigns, which is summarized as follows:

\begin{small}
\begin{tcolorbox}
\vspace{-0.5\baselineskip}
\textbf{Ping-IE: Hijack impact estimation with ping campaigns}
\begin{enumerate}[leftmargin=*]
\item Select randomly a set of $M$ ASes.
\item For each AS, select a set of $N_{IP}$ responsive (``pingable'') IP addresses, ping them, and monitor for the replies. 
\item 
If at least one IP address of an AS $i$ replies to the ping, then set $\hat{m}_{i}=0$, otherwise set $\hat{m}_{i}=1$.
\item Estimate the hijack impact from the NIE expression in \eq{eq:definition-nie}, by using $\hat{m}_{i}$ instead of $m_{i}$.
\end{enumerate}
\vspace{-0.8\baselineskip}
\end{tcolorbox}
\end{small}
Despite the simplicity of the Ping-IE steps, the accuracy heavily depends on the parameters $M$ and $N_{IP}$, and the set of ``pingable'' IP addresses. In the remainder, we study the expected accuracy and how to carefully tune these parameters.

\subsection{The Effect of Failed Measurements on the NIE}
Let assume that we conduct ping measurements to an IP address in the AS $i$ to infer if it is infected (no ping reply received) or not (ping reply received). However, if AS $i$ is not infected, but the selected IP address is configured to not reply to pings, or for some other reason unrelated to the hijack a ping reply never reaches our system, then we will incorrectly infer that the AS $i$ is infected. If this happens with several ASes, the NIE will overestimate the hijack impact.

The following theorem quantifies the introduced bias (i.e., overestimation of hijack impact) and the RMSE of the NIE as a function of the measurement failure probability.

\begin{definition}[Measurement failure probability]\label{def:ping-failure}
Let $m$ be an indicator that denotes whether a monitor is infected ($m=1$) or not ($m=0$), and $\hat{m}$ be its measured value. The measurement failure probability, $p$, is the probability of measuring as infected a non-infected monitor, i.e., 
\[p = Prob\{\hat{m}=1|m=0\}\]
\end{definition}

\begin{theorem}\label{thm:NIE-bias-rmse-p}
Under a randomly selected set of monitors $\mathcal{M}$, and a measurement failure probability $p$, it holds for NIE:
\begin{align*}
& Bias_{NIE} = c_{I}^{'}\cdot p\\
& RMSE_{NIE} = \int_{0}^{1} \sqrt{\textstyle\frac{A_{I,p}}{M} + B_{I,p}}\cdot f(I)\cdot dI
\end{align*}
where $c_{I}^{'} = 1-E[I]$ is a constant that depends on the impact distribution $f(I)$, and $A_{I,p}$ and $B_{I,p}$ are given by
\begin{align}
A_{I,p} &= \left(I+(1-I)\cdot p\right)\cdot(1-I)\cdot (1-p)\nonumber\\
B_{I,p} &= (1-I)^{2}\cdot p^{2}\nonumber
\end{align}
\end{theorem}
\begin{proof}
The proof is given in Appendix~\ref{appendix:proof-NIE-bias-rmse-p}.
\end{proof}

\begin{corollary}\label{corollary:rmse-inf}
$RMSE_{NIE} \geq RMSE_{NIE}(M\rightarrow \infty) = c_{I}^{'}\cdot p$.
\end{corollary}

The values $c_{I}^{'}=1-E[I]$ for the different hijack types and \{V,H\} pairs can be calculated from Table~\ref{tab:hijack-constant}. For example, for random \{V,H\} pairs and hijack Types-0, 1, and 2, the constant $c_{I}^{'}$ is 0.5, 0.7, and 0.81, respectively.
\textit{Remark:} It is interesting to note that the RMSE of NIE increases with the hijack type when $p>0$, whereas the opposite holds when $p=0$ (see~\secref{sec:understanding}). This is due to the fact that $p$ affects only the monitors that are \textit{not} infected, and thus in cases where the impact is lower, i.e., for higher hijack types, the error due to ping failures is higher.




\myitemit{Key finding:} \textit{When the failure probability is larger than 20\% ($p\geq 0.2$), ping campaigns --no matter how many ASes are pinged-- are less accurate than NIE with RC or RA monitors.} 

As expected, NIE under ping failures becomes a biased estimator. Also its RMSE increases with the failure probability $p$, which means that for high $p$ the NIE with ping campaigns becomes worse than the NIE with public services. To better quantify the expected estimation error, we present in Fig.~\ref{fig:sims-theory-rmse-nie-p-type-0} the RMSE for hijacks of Type-0 (for Types $\geq1$ the RMSE is higher) as a function of the failure probability $p$ (x-axis) and the number of monitors $M$ (i.e., pinged ASes).

\begin{figure}
\centering
\subfigure[]{\includegraphics[width=0.49\linewidth]{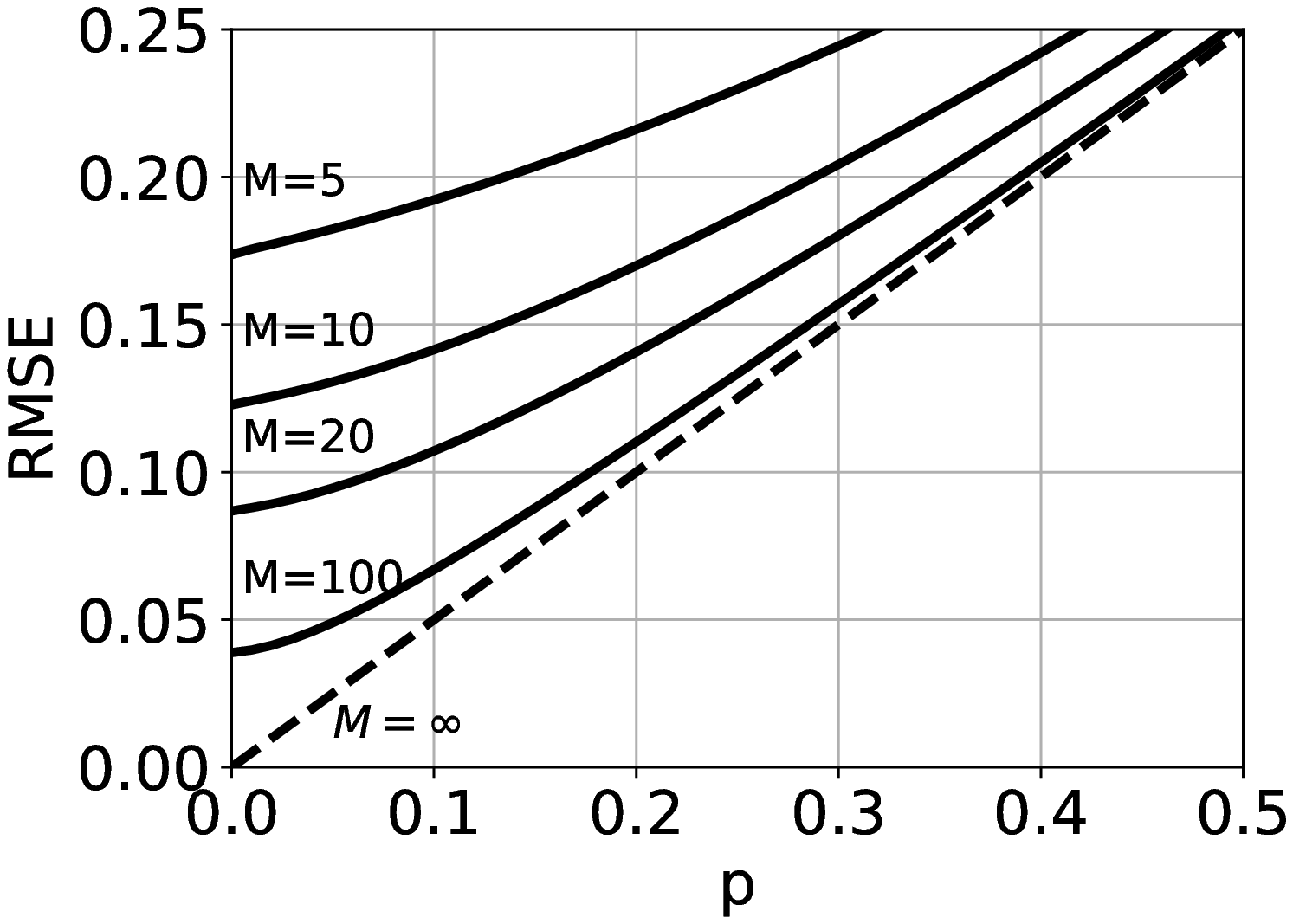}\label{fig:sims-theory-rmse-nie-p-type-0}}
\subfigure[]{\includegraphics[width=0.49\linewidth]{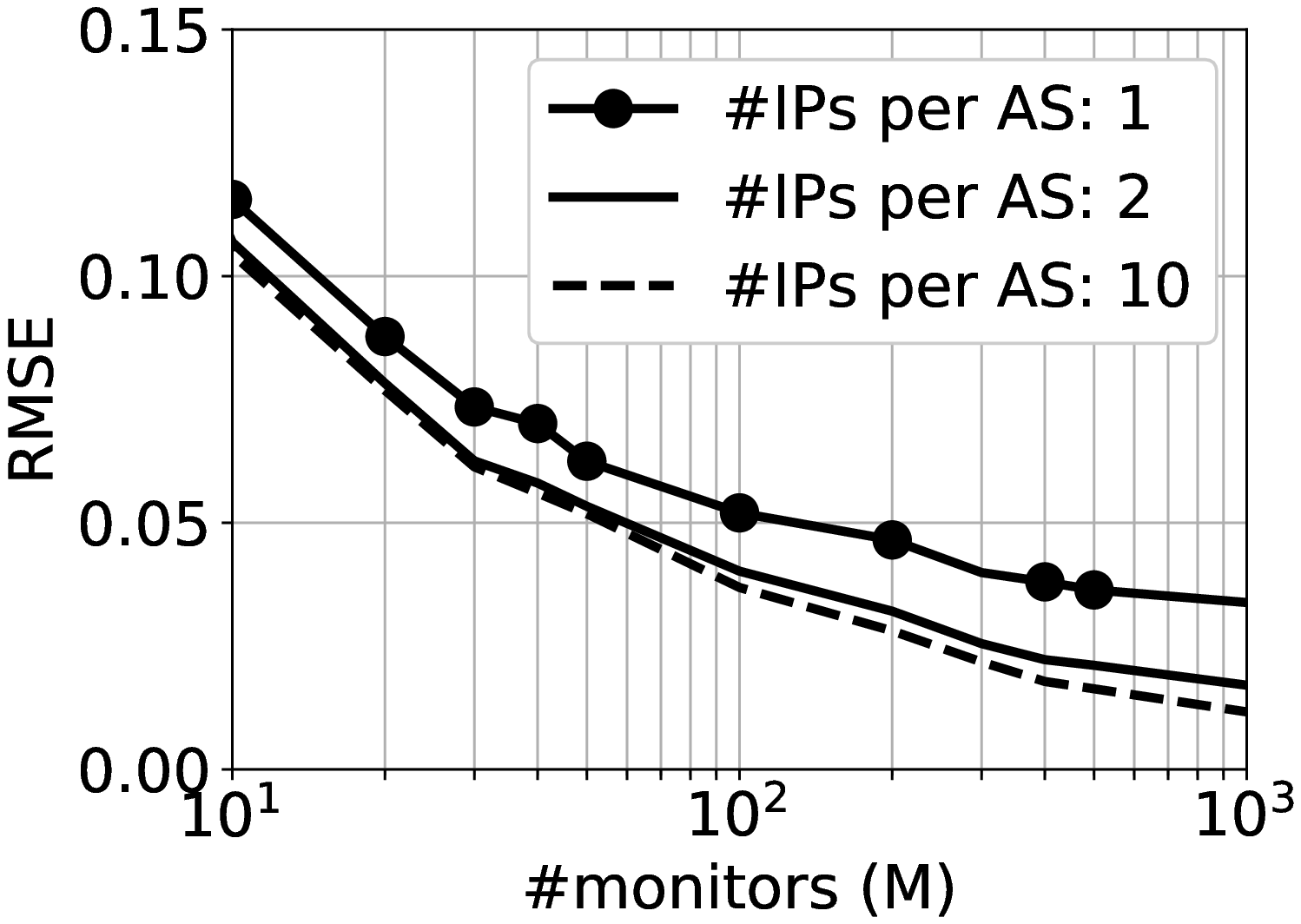}\label{fig:rmse-nie-pg-vs-nb-pings-experiments}}
\vspace{\myfigureskip}
\caption{RMSE (y-axis) of the NIE / Ping-IE under random set of $M$ monitors: (a) theoretical results for the NIE ($N_{IP}=1$) under different ping failure probabilities $p$ (x-axis); (b) results from the PEERING experiments for the Ping-IE under different number of monitors (x-axis) and $N_{IP}$ (legend).
}
\label{fig:nie-ping}
\end{figure}

Based on the analytical findings, we proceed to design and fine tune a methodology by considering practical issues.

\vspace{-0.5\baselineskip}
\subsection{Practical Design of the Ping-IE}
The above results indicate that for a low RMSE, we need to have low ping failure probabilities $p$. One can achieve this by (i) carefully selecting the set of IPs to be pinged so that they are ``pingable'', and/or (ii) pinging more than one IPs per AS and waiting for a reply from at least one IP. In the methodology we propose, we do both. 
Specifically, we first quantify the ping failure probability we expect to have in practice, and based on this, we select the set of IP addresses (which IP addresses and how many per AS) to be pinged.

\myitem{Selection of pingable IP addresses.}  If we select arbitrarily an IP address within the prefixes of an AS, then the failure probability is very large (more than 90\%~\cite{pingable-ip-2013,heidemann2008census})
; this would lead to a very inefficient methodology with $RMSE>45\%$. Therefore, we use a list of IP addresses provided by ANT Lab~\cite{antlabhitlists} that have high probability of replying to pings. Specifically, we compile a list of pingable IP addresses per AS, by combining the following two datasets:
\begin{itemize}[leftmargin=*]

\item {\textit{IP hitlists from ANT Lab}}~\cite{antlabhitlists}: These are lists of IP addresses that are found to be reachable via \texttt{ping} with high probability, based on past measurements. We select only the IP addresses with more than 90\% confidence score.

\item {\textit{Prefix-to-AS mapping:}} We consider information from RIPE RIS~\cite{riperis} (via the RIPEstat API~\cite{ripestat}) and RouteViews~\cite{routeviews} (via CAIDA's \texttt{pfx2as}~\cite{caidapfx2as}). These datasets map IP prefixes announced in BGP to the originating ASNs. We filter out mappings that are inconsistent within a period of a month (e.g., due to transient incidents), and merge the two datasets.
\end{itemize}


\myitem{Selection of the number of IP addresses to ping per AS ($N_{IP}$).} Pinging more than one IP addresses per AS, increases the probability to obtain a correct inference about whether it is infected (i.e., we need at least one ping reply to infer an AS as non-infected
). But, \textit{how many IP addresses need to be pinged per AS to have a low error}? 
To quantify this, we conduct a set of measurements using the PEERING testbed: we announce a prefix from PEERING, ping from a host within PEERING the top 10 pingable IP addresses per AS 
for $\sim46k$ ASes, and monitor which IP addresses reply to the pings.

\myitemit{Key finding:} \textit{To achieve low estimation error we need to ping at least 2 IP addresses from the ANT Lab's IP hitlists~\cite{antlabhitlists} per AS. Pinging more than 3 IP addresses, does not significantly improve accuracy.}

Table~\ref{tab:ping-failure-probability} (top row) presents the fraction of ASes for which we did not receive any reply from pinging their top-$x$ ($x=1,...,10$) IP addresses. We can see that the failure probability is quite high (12.8\%) when pinging only one IP address per AS, which indicates that more measurements per AS are needed to enable an accurate impact estimation. With 3 measurements the ping failure probability decreases to 2.1\% and further decreases gradually to 0\% with 10 pings per AS. The middle rows of Table~\ref{tab:ping-failure-probability} show the \textit{lower bound} for the RMSE ($M\rightarrow\infty$) that can be achieved by Ping-IE for different hijack types, which indicates that only one ping per AS may not be enough to outperform NIE with public monitors. 

Finally, the bottom rows give the RMSE for Type-0 hijacks and practical values of $M$, calculated from the expressions of Theorem~\ref{thm:NIE-bias-rmse-p} (using the values $p$ that correspond to the $N_{IP}$ from the first row of the table). While pinging only 10 ASes is not efficient, pinging 100 ASes already achieves an accuracy relatively close to the best achievable ($M\rightarrow\infty$).

\begin{table}[h]
\caption{Top row: Probability of ping failure for ANT Lab's IP hitlists~\cite{antlabhitlists} (PEERING experiments). Middle/Bottom rows: RMSE of the Ping-IE vs. nb. of pinged top IPs per AS (theory).}
\label{tab:ping-failure-probability}
\centering
\begin{small}
\begin{tabular}{ll|ccccccc}
{}&{} & \multicolumn{5}{c}{\# of pinged IP addresses per AS}\\
{}&{}&  {1}  & {2}   & {3} &{...}&{10}\\
\hline
\multicolumn{2}{l|}{\% ASes with no reply}&  {12.8\%}  & {~4.2\%}   & {~2.1\%} &{...}&{0\%}\\
\hline
\multirow{ 3}{*}{\shortstack[l]{RMSE\\M=$\infty$}}&  Type-0 & {6.4\%} & {2.1\%} & {1.0\%}  & {...} &{0\%}\\
&Type-1& {9.0\%} & {3.0\%} & {1.4\%}  & {...} &{0\%}\\
&Type-2& {10.4\%} & {3.4\%} & {1.7\%}  & {...} &{0\%}\\
\hline
\multirow{ 3}{*}{\shortstack[l]{RMSE\\ Type-0}}& M=10 & {14.9\%} & {12.9\%} & {12.6\%}  & {...} &{12.3\%}\\
{}&M=50 & {9.0\%} & {6.2\%} & {5.7\%}  & {...} &{5.5\%}\\
{}&M=100 & {7.9\%} & {4.7\%} & {4.1\%}  & {...} &{3.9\%}\\
\end{tabular}
\end{small}
\end{table}

In Fig.~\ref{fig:rmse-nie-pg-vs-nb-pings-experiments} we present the corresponding results from the real experiments, which are in agreement with the main theoretical findings\footnote{For small $N_{IP}$ or $M$, the RMSE values in our experiments are a bit lower than the corresponding theoretical values (Table~\ref{tab:ping-failure-probability}); this is due to the small number of experiments (confidence intervals are larger for small $N_{IP}$ or $M$)
.}. In particular, we observe that pinging 2 IP addresses per AS ($N_{IP}=2$), significantly reduces the RMSE compared to the case of $N_{IP}=1$. However, the improvement by further increasing the $N_{IP}$ (up to $N_{IP}=10$) is marginal. Moreover, in all $N_{IP}$ cases, increasing the number of pinged ASes more than $M=500$ barely improves accuracy (as was already indicated by Fig.~\ref{fig:sims-theory-rmse-nie-p-type-0}).

\section{Improving Impact Estimation with Public Infrastructure Monitors}\label{sec:ridge-estimator}
\myitem{The biased view of public monitors.} As already discussed
, public monitors are not uniformly deployed in the Internet, and this increases the error of NIE. Consider the following example scenario: The victim is an AS that is located in a geographical area where many monitors exist (e.g., with direct peering or short paths to these monitors), and the hijacker AS is in a different area with less monitors. The actual impact of the hijacks is $I=50\%$ (i.e., half of all ASes are infected), however, the NIE underestimates the impact (i.e., $\hat{I}_{NIE}<I$) because more monitors would prefer the paths to the victim.

Generalizing the above example, the error of NIE increases when the monitors are not representative of the global connectivity, or --more abstractly-- when there are correlations between their measurements (due to locations, underlying topology, AS-relationships, etc.). Hence, \textit{to improve the estimation accuracy of NIE under public monitors, one needs to take into account the correlations between the monitors}. To this end, in the following we design a statistical learning methodology that exploits information of past events (to identify correlations between public monitors), fits a model that diminishes the effect of correlations, and returns an estimation for the impact. 

\myitem{The linear regression estimator (LRE).} The methodology we propose is summarized as follows:

\begin{small}
\begin{tcolorbox}
\vspace{-0.5\baselineskip}
\textbf{LRE: Linear Regression Estimator}
\begin{enumerate}[leftmargin=*]
\item Compile a dataset from $N$ past events (hijacks, anycast announcements, etc.), where for each event $j$, $j=1,...,N$, collect the measurements $m_{i}^{(j)}$ of the monitors $i\in\mathcal{M}$, and the actual hijack impact $I^{(j)}$.
\item Fit a least squares estimator, by calculating the weights $w_{i}$, $i\in\mathcal{M}$, as:
\begin{equation*}
\textstyle \mathbf{w}\leftarrow \arg\min_{\mathbf{w}} ~\left(||\mathbf{M}\cdot \mathbf{w}- \mathbf{I}||_{2}\right)^{2} + \alpha\cdot \left(||\mathbf{w}||_{2}\right)^{2}
\end{equation*}
where $\mathbf{w} = [w_{1}, ..., w_{M}]$ and $\mathbf{I}=[I^{(1)}, ...,I^{(N)}]$ are vectors, $\mathbf{M}$ is the matrix with elements $m_{i}^{(j)}$ at the $i^{th}$ row and $j^{th}$ column, and $||\cdot||_{2}$ denotes the l2-norm.

\item Estimate the hijack impact from the current monitor measurements $m_{i}$ and the calculated weights $w_{i}$ as
\begin{equation*}
\hat{I} = \textstyle\sum_{i\in\mathcal{M}} m_{i}\cdot w_{i}
\end{equation*}
\end{enumerate}
\vspace{-0.8\baselineskip}
\end{tcolorbox}
\end{small}

The first step is to collect data that contain information about the correlations between the measurements of the different public monitors. To do this, one can consider a set of past/ongoing events $\mathcal{N}$ ($|\mathcal{N}|=N$), where two (at least) ASes announce the same prefix. Such events can be actual or emulated (e.g., as in our experiments;~\secref{sec:real-world-exp}) hijacking events, or legitimate anycasting announcements (which from a routing point of view are equivalent to Type-0 hijacks)~\cite{de2017broad}. For each of these events $j\in\mathcal{N}$, we collect the measurements $m_{i}^{(j)}$ of the public monitors $i\in\mathcal{M}$. In the case of RC monitors the measurements can be retrieved from the the RIPE RIS~\cite{riperis} and the RouteViews~\cite{routeviews} services directly, or from the open-source tool BGPStream~\cite{bgpstream-website, bgpstream-paper} that aggregates these measurements, and in the case of RA from the RIPE Atlas API (by triggering measurements for ongoing events, or collecting the periodic measurements for past events)~\cite{ripeatlas}. Moreover, for each event we need to know the actual hijack impact
, which can be exactly measured with exhaustive ping measurements (similarly to our methodology in \secref{sec:real-world-exp} for collecting the ground-truth in our experiments, or other related approaches~\cite{de2017broad}) or approximated well with a few thousands of ping measurements using the methodology of \secref{sec:ping-estimator}. 

The second step is to identify any correlations between the measurements of the monitors, and eliminate their effect in the estimation. We select to do this, by using a least squares approach, and, in particular, a linear regression estimator with regularization of the weights (i.e., Ridge regression). Our choice, is motivated by the fact that (i) the least square estimator (i.e., linear regression) has the the lowest variance within the class of linear unbiased estimators~\cite{davidson2000econometric}\footnote{We tested several non-linear estimators as well 
(e.g., support-vectors, random forests, neural networks). However, they had similar (or worse) performance to LRE. We selected the LRE, as a simple model, which comes with the advantages of interpretability, need for less training data, 
etc.}, and (ii) the regularization significantly reduces the variance of the estimations when multi-collinearity occurs;  in fact, the public monitor measurements are highly collinear, and thus large values of the regularization parameter $\alpha$ are needed (e.g, we found that values $a\geq 50$ performed best)

Finally, having fitted the model (i.e., the weights $w_{i}$) that eliminates the correlations between monitor measurements, we can apply it to any new hijacking event and estimate its impact.

\myitem{LRE vs. NIE estimation accuracy.} We compare the accuracy of the impact estimations by LRE and NIE with RA monitors in Fig.~\ref{fig:RMSE-ridge-vs-nb-monitors} (similar results hold for the RC monitors). We use $1000$ simulations as the past-events dataset, from which we collect the data to fit the LRE, and apply the LRE and NIE to a different set of $1000$ simulations.

\myitemit{Key finding:} \textit{LRE can eliminate the effect of correlations in public monitor measurements and achieve similar efficiency to the (best performing) ping-based estimators.}

We see that LRE has significantly lower RMSE than NIE. In fact, in the case of Type-0 hijacks (Fig.~\ref{fig:RMSE-ridge-RA-sims-type0}) LRE achieves equal accuracy to the NIE with random monitors (Theorem~\ref{thm:NIE-bias-rmse}), or even better accuracy for small number of monitors. This is an important finding that demonstrates that we can design estimators based on public monitors with similar efficiency to the ping-based estimators. LRE outperforms NIE for the case of Type-2 hijacks as well (Fig.~\ref{fig:RMSE-ridge-RA-sims-type2}), e.g., having almost 50\% less RMSE for $M\geq 50$ monitors. Comparing the RMSE of LRE in the cases of Type-0 and Type-2 hijacks, we can see that it increases with the hijack type; this is due to the fact that the actual impact of higher type hijacks is lower, and thus there are more observations $m_{i}=0$, which makes more difficult for a model to identify the existing correlations in measurements.

\begin{figure}
\centering
\subfigure[Type-0 hijack]{\includegraphics[width=0.49\linewidth]{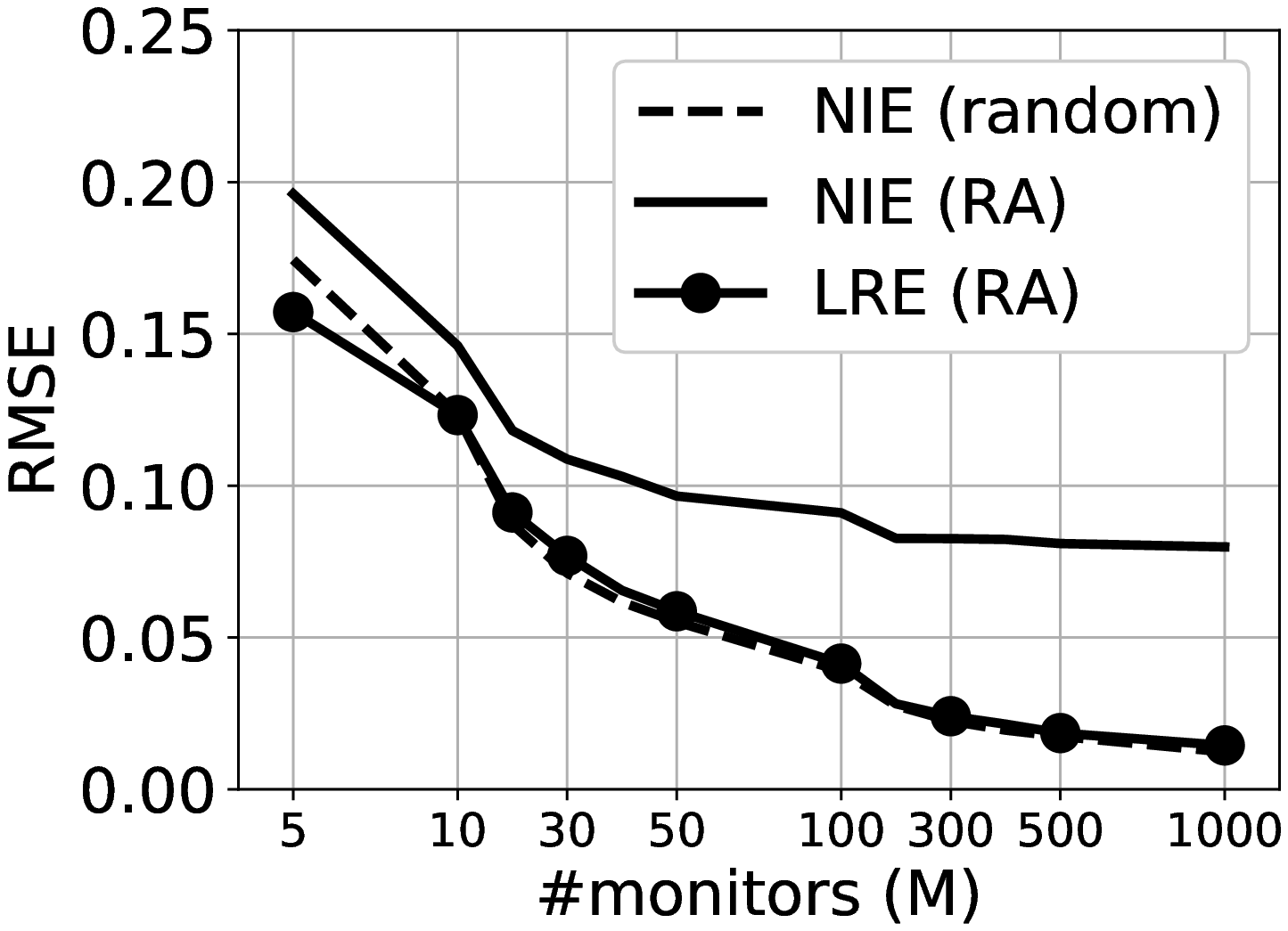}\label{fig:RMSE-ridge-RA-sims-type0}}
\subfigure[Type-2 hijack]{\includegraphics[width=0.49\linewidth]{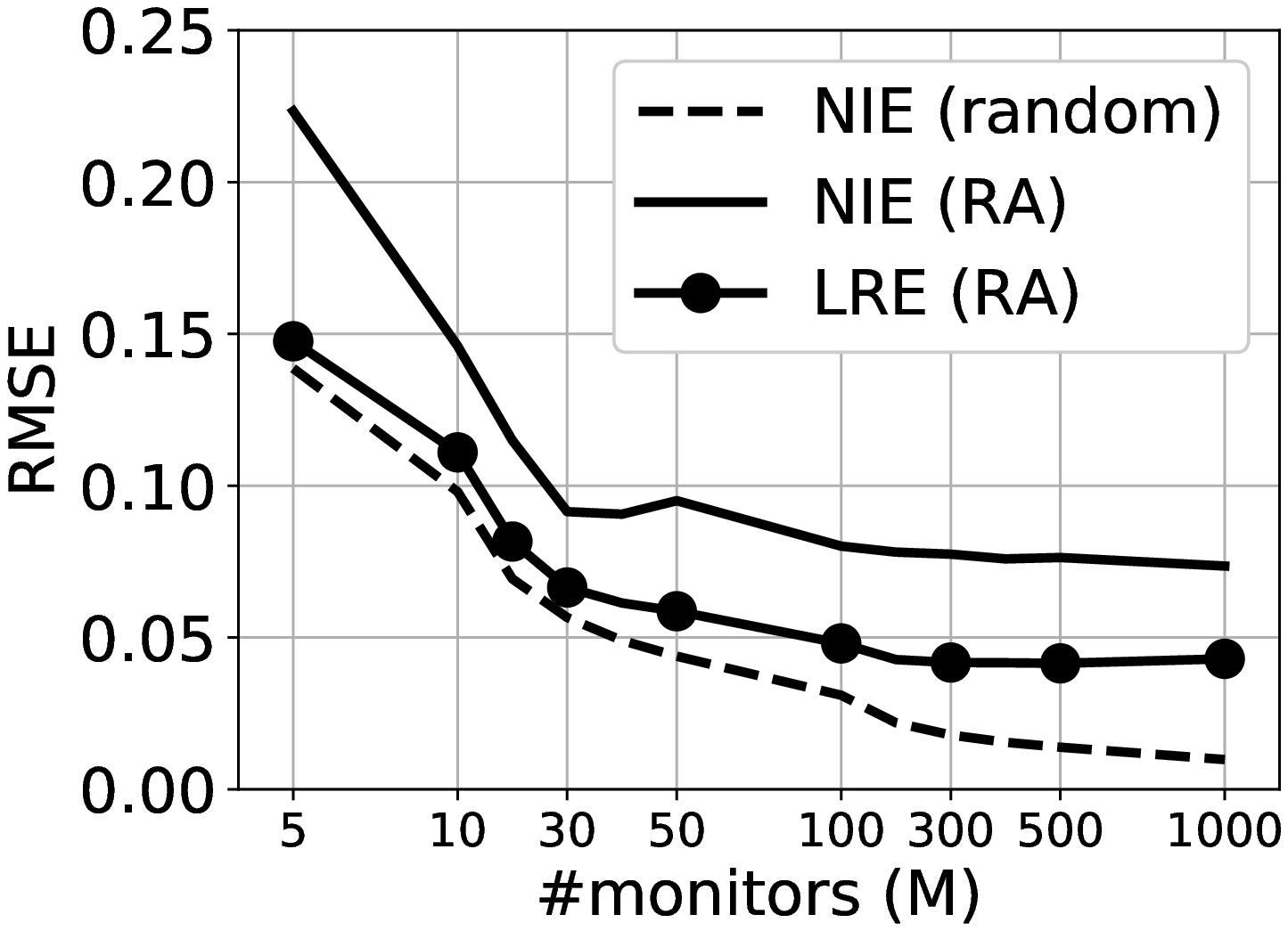}\label{fig:RMSE-ridge-RA-sims-type2}}
\vspace{\myfigureskip}
\caption{RMSE (y-axis) of the LRE with RA monitors and the NIE with random or RA monitors, vs. number of monitors (x-axis), in simulations of (a) Type-0 and (b) Type-2 hijacks.}
\label{fig:RMSE-ridge-vs-nb-monitors}
\end{figure}

We proceed to test the efficiency of LRE in the real experiments with PEERING. We remind that we have only 22 experiments, which is a very small dataset for training a model. Hence, this is not a conclusive evaluation (it can be rather seen as a stress-test of LRE); nevertheless, our findings are very promising. For each experiment \{V,H\}, we consider as the dataset with past events the other 20 experiments (omitting also the experiment \{H,V\}) and fit the LRE
. Figure~\ref{RMSE-ridge-RA-PEERING} presents the results for RC and RA monitors. In the case of RA monitors (Fig.~\ref{fig:RMSE-ridge-RA-PEERING}), we can see that LRE clearly outperforms NIE with RA monitors, despite the very limited available data for training the LRE. In the case of RC (Fig.~\ref{fig:RMSE-ridge-RC-PEERING}), LRE has a similar performance to NIE with RC monitors. These findings verify the efficiency of LRE, and indicate that even information from only a few past events can lead to more accurate estimations. However, they also highlight the importance of collecting past event data for training estimators; our experiments data~\cite{sermpezis2021estimating-code} can contribute to this direction.

\begin{figure}
\centering
\subfigure[RA monitors]{\includegraphics[width=0.49\linewidth]{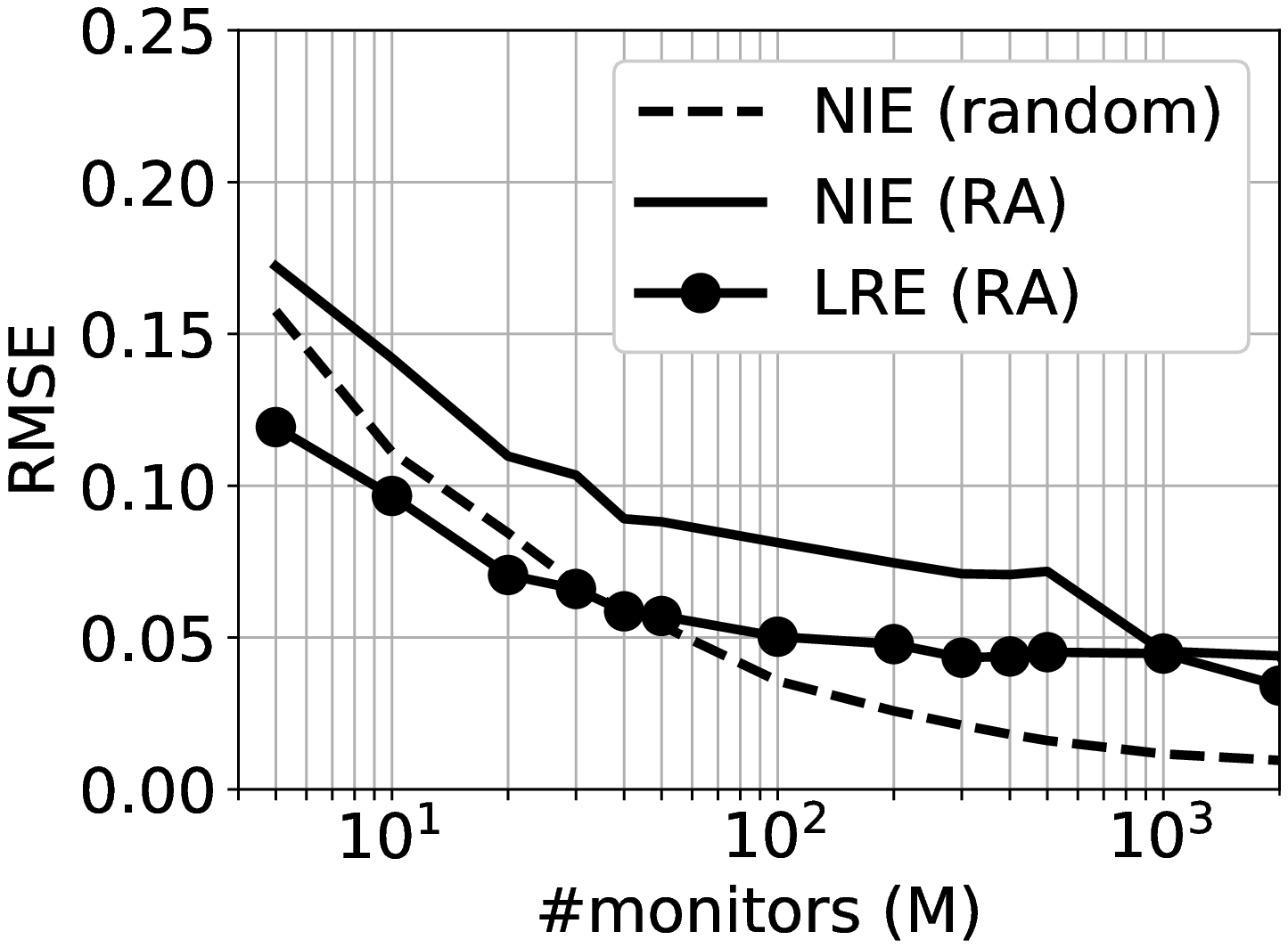}\label{fig:RMSE-ridge-RA-PEERING}}
\subfigure[RC monitors]{\includegraphics[width=0.49\linewidth]{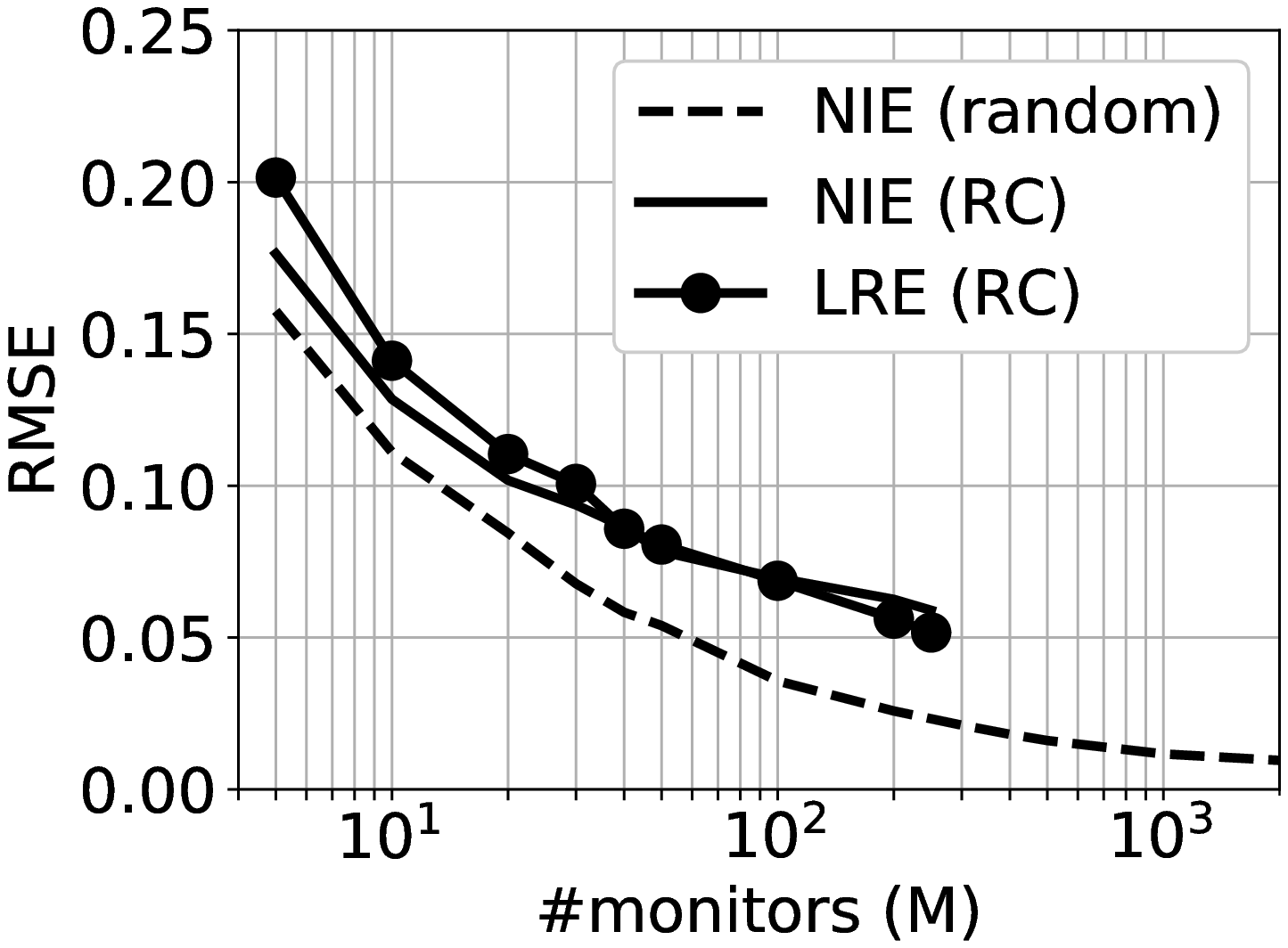}\label{fig:RMSE-ridge-RC-PEERING}}
\vspace{\myfigureskip}
\caption{RMSE (y-axis) of the LRE with public monitors and the NIE with random or public monitors, vs. number of monitors (x-axis) for Type-0 hijacks for (a) \textit{RA} and (b) \textit{RC} monitors in the PEERING experiments.}
\label{RMSE-ridge-RA-PEERING}
\end{figure}

\section{Related Work}\label{sec:related}
The majority of works on BGP prefix hijacking (or other types of events affecting the Internet operation, e.g., outages~\cite{quan2013trinocular,padmanabhan2019find}) focus on the detection of an event, using network measurements on the control plane~\cite{sermpezis2018artemis
} or the data plane~\cite{Zhang-Ispy-CCR-2008
} or both~\cite{heap-jsac2016,
Shi-Argus-IMC-2012
}. The difference between \textit{detection} and \textit{impact estimation} lies in the fact that having information for at least one infected AS 
is typically enough to enable the detection of a hijack, however, it gives only limited information (if any) about its overall impact. Our work focuses on quantifying the impact of a (detected) hijack, and thus complements the existing detection methods.

A few works studying (through simulations) the hijack impact focus mainly on the \textit{average} impact of different hijacking attacks~\cite{ballani2007study} or hijacker ASes~\cite{qiu2010towerdefense,Lad2-Understanding-resiliency-hijacks-DSN-2007}, whereas our goal is to estimate (through measurements) the \textit{actual impact of an ongoing hijack}. Ballani et al.~\cite{ballani2007study} consider interception hijacks
, and study when they are expected to have significant impact, by providing coarse estimates for groups of ASes (e.g., Tier-1) that could act as hijackers. Similarly, the potential impact of a hijacker AS (or, conversely, the resilience of a victim AS to a hijack) is studied in~\cite{Lad2-Understanding-resiliency-hijacks-DSN-2007}, where the topological characteristics (e.g., node degree) of ASes are used to classify potential hijackers based on the impact they can cause. TowerDefense~\cite{qiu2010towerdefense} aims to find a set of monitors that maximizes the probability to detect a hijack (i.e., at least one monitor is infected); a problem that is complementary to impact estimation, whose aim is to find a \textit{representative set of monitors} (i.e., a set where the fraction of infected monitors is close to the overall hijack impact).

Finally, the framework of~\cite{sermpezis2019pomacs} for predicting the catchment of an anycast deployment, could be used for hijack impact estimation (a hijack can be seen as a setup where the pair \{V,H\} anycasts the same prefix). However, the routing information that is required may not be always accurately known in practice, which would lead to higher errors than a measurement-based NIE (we verified this in our experiments).

\section{Conclusion}\label{sec:conclusion}

The problem of hijack impact estimation has not been given attention in literature, despite its usefulness for network operations and economy, e.g., to know how an ongoing hijack affects a network, or to select and evaluate the efficiency/cost of different mitigation measures. In this paper we made the first steps towards understanding the fundamental (limits, trade-offs, etc.) and practical aspects (use of public infrastructure, measurement failures, etc.) of the impact estimation problem. We also designed accurate estimation techniques that are easy to implement and incorporate in existing defense systems. 



We believe that this work can motivate further research on the topic; we identify and discuss two interesting directions:

A network may exchange traffic with only a subset of ASes in the Internet (and/or different volumes of traffic per AS). In this case, a more fine-tuned estimation of the impact (on the exchanged traffic) can further help network operators. Sophisticated estimators, such as weighted versions of the NIE, Ping-IE, or LRE, can be designed. A preliminary view in this direction is given in Fig.~\ref{fig:RMSE-traffic-subset}, where we apply NIE to estimate the impact only on subsets of ASes: the RMSE of NIE with random set of monitors remains almost constant (we observed similar behavior for the NIE with RC and RA monitors as well). While a formal and detailed study is needed to draw firm conclusions, this is an indication that the main findings of this paper can hold for more generic cases as well. 

The LRE findings can motivate research on the selection and combination of measurements from public infrastructure, with broader applications on the Internet monitoring. To this end, Fig.~\ref{fig:correlations-LRE-weights} provides some initial statistics on the correlation between the LRE weights $w_{i}$ (i.e., the importance of monitor $i$) and the topological characteristics of monitors $i$. Admittedly correlations are weak, however, they reveal some underlying trends and give rise to some interesting questions: (i) Observations from RC and RA monitors/ASes with larger customer cones seem to play a more important role; does this indicate that we should deploy more monitors on such networks? (ii) Top tier networks (e.g., Tier-1) contribute more on the LRE, but only in the RC case; should we devise different strategies for selecting measurements from RC and RA? (iii) Finally, the number of neighbors a monitor has, seems to be a less important feature; would this mean that monitors at IXPs (where networks establish a lot of peerings) are equally important with monitors at stub networks? Further research and (open) data could help answering such questions.

\begin{figure}
\centering
\subfigure[Traffic subset]{\includegraphics[width=0.49\linewidth]{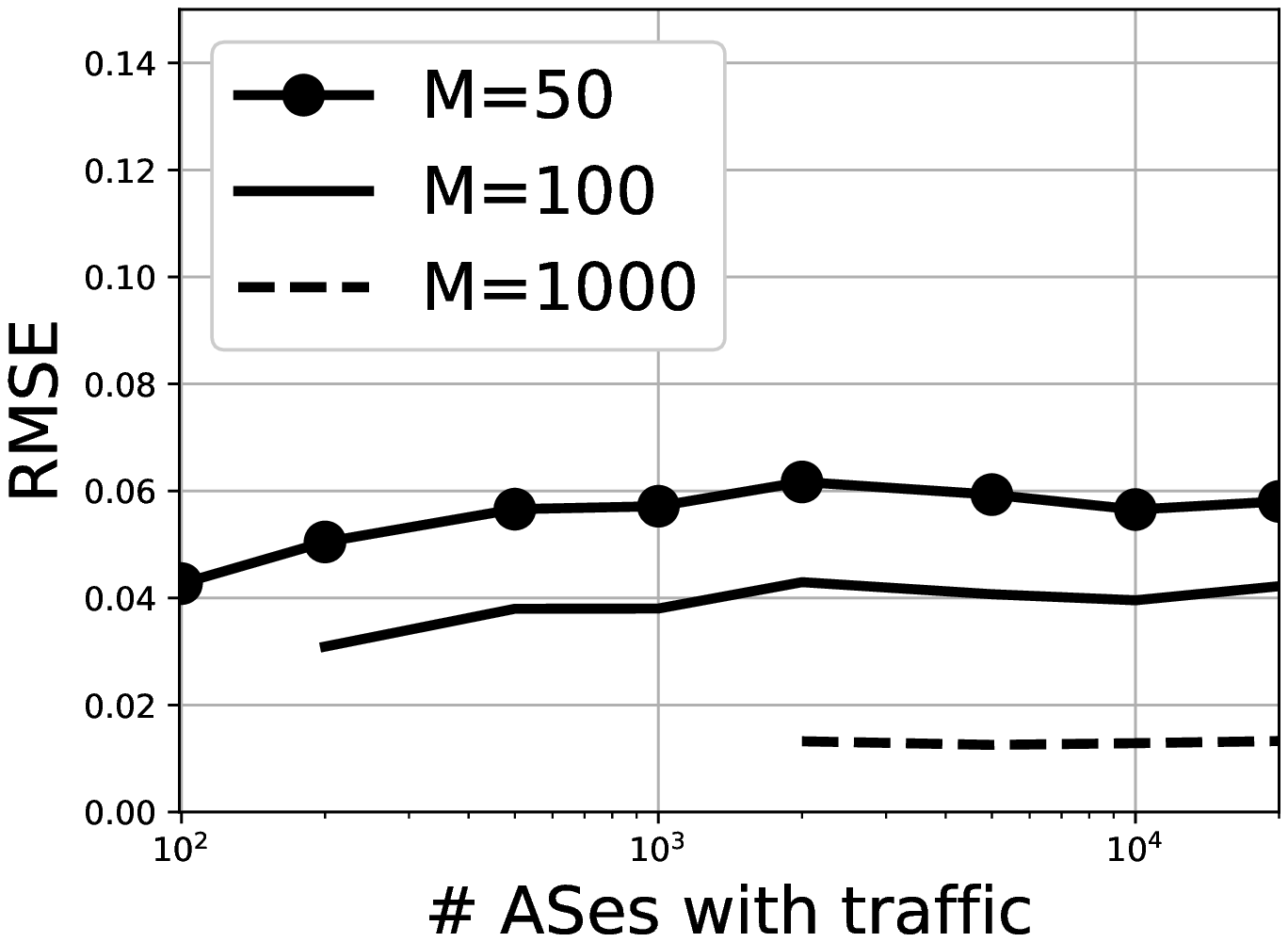}\label{fig:RMSE-traffic-subset}}
\subfigure[LRE weights]{\includegraphics[width=0.49\linewidth]{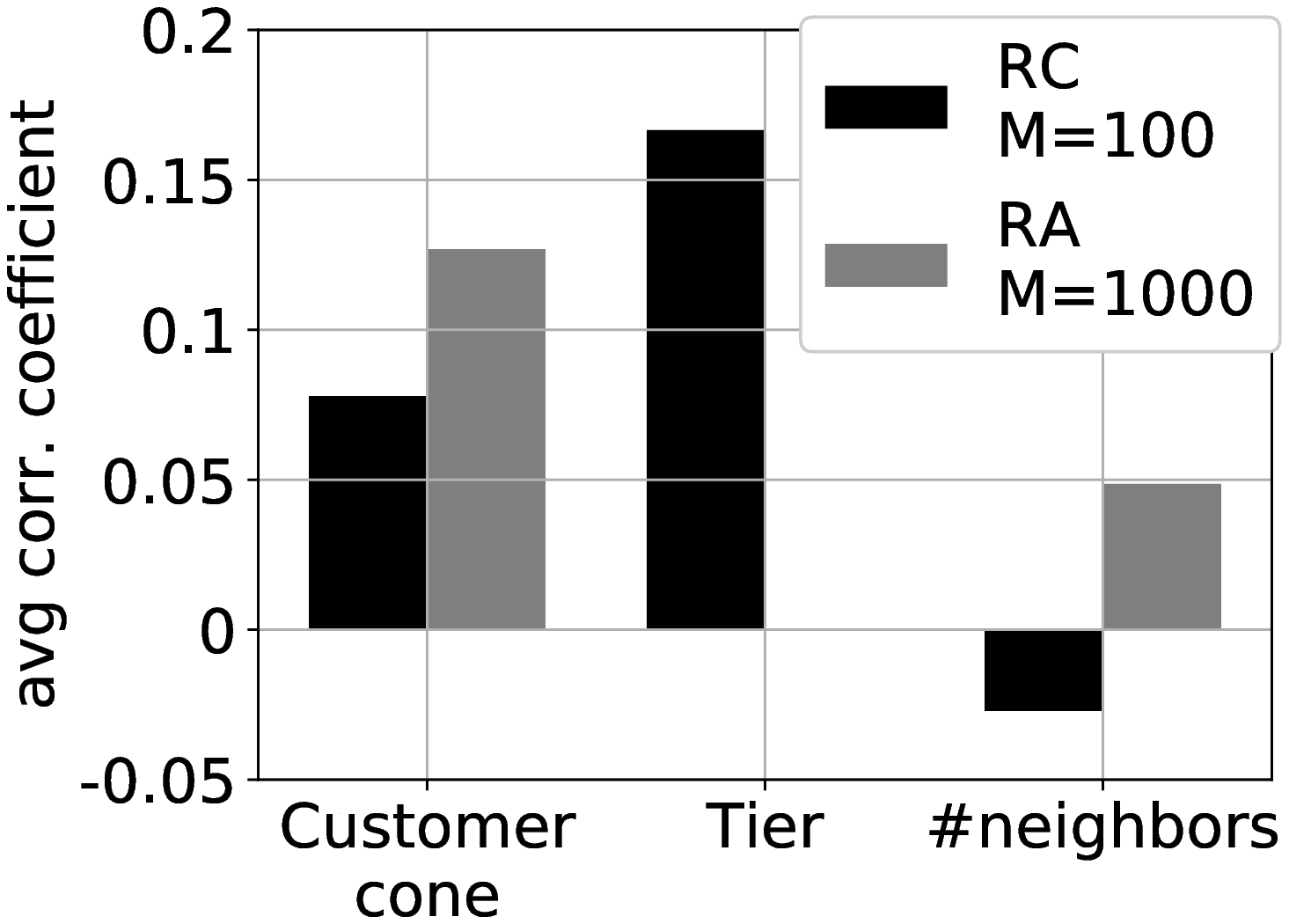}\label{fig:correlations-LRE-weights}}
\vspace{\myfigureskip}
\caption{(a) RMSE (y-axis) of NIE with random set of $M$ monitors measuring the hijack impact on different number of ASes with traffic to the victim AS (x-axis).  (b) Correlation between LRE weights $w_{i}$ and AS-topology characteristics of monitors $i$; average values over 30 simulation scenarios.}
\end{figure}


{ \balance
{
\bibliographystyle{IEEEtran}
\bibliography{references,artemis}
}
}

\newpage
\appendices
\section{Proof of Theorem~\ref{thm:NIE-bias-rmse}}\label{appendix:proof-NIE-bias-rmse}

The hijack impact $I$ is the fraction of ASes that are infected. 
When choosing randomly the ASes/monitors to be measured (and the number of measurements is much less than the total number of ASes), we can reasonably model each measurement as an independent Bernoulli trial with probability $I$, i.e., the probability that a selected monitor is infected is $I$. 
%
The NIE involves the sum of $M$ independent Bernoulli trials
, and thus the sum follows a binomial distribution
, i.e., $\sum_{i\in\mathcal{M}}m_{i}\sim Binomial(M,I)$,
for which it holds that
\begin{align}
E\left[\textstyle\sum_{i\in\mathcal{M}}m_{i}\right] &= M\cdot I\label{eq:expectation-binomial}\\
Var\left[\textstyle\sum_{i\in\mathcal{M}}m_{i}\right] &= M\cdot I\cdot (1-I)\label{eq:variance-binomial}
\end{align}

\myitem{Bias.} The bias of an estimator is defined as
\begin{equation}\label{eq:bias-generic}
Bias = E[\hat{I}-I] = E[\hat{I}]-I
\end{equation}
For the NIE it holds that 
\begin{align}
\textstyle E[\hat{I}]  
= E\left[\frac{1}{M}\cdot \sum_{i\in\mathcal{M}}m_{i}\right]
= \frac{1}{M}\cdot E\left[ \textstyle\sum_{i\in\mathcal{M}}m_{i}\right]
=I\label{eq:expectation-nie}
\end{align}
where in the last equation we used the expression of \eq{eq:expectation-binomial}. Substituting \eq{eq:expectation-nie} in \eq{eq:bias-generic} gives $Bias=0$.

\myitem{RMSE.}
The RMSE of NIE, given the actual impact $I$, is:
\begin{align}
\textstyle
RMSE(I)   
&\textstyle = \sqrt{E[(\hat{I} - I)^2]} 
= \sqrt{E[(\hat{I} - E[\hat{I}])^2]}\nonumber\\
&\textstyle = \sqrt{Var[\hat{I}]}
= \sqrt{Var[\frac{1}{M}\textstyle\sum_{i\in\mathcal{M}}m_{i}]}\nonumber\\
&\textstyle = \sqrt{\frac{1}{M^{2}}\cdot Var[\textstyle\sum_{i\in\mathcal{M}}m_{i}]}
= \frac{\sqrt{I\cdot (1-I)}}{\sqrt{M}}\label{eq:rmse-i}
\end{align}
where we use the definition of the variance ($Var(x)=E[(x-E[x])^{2}]$), and in the last equality the expression from \eq{eq:variance-binomial}.

Then, the RMSE of NIE follows by taking the expectation of $RMSE(I)$ in \eq{eq:rmse-i} over the impact distribution $f(I)$:
\begin{align}
\textstyle
RMSE    
= \int_{0}^{1} \frac{\sqrt{I\cdot (1-I)}}{\sqrt{M}} \cdot f(I)\cdot dI
= \frac{1}{\sqrt{M}} \cdot c_{I}
\end{align}



\section{Proof of Theorem~\ref{thm:NIE-bias-rmse-p}}\label{appendix:proof-NIE-bias-rmse-p}

If $m_{i}=1$, then $\hat{m}_{i}$ is $1$ as well. However, if $m_{i}=0$, then $\hat{m}_{i}$ is $1$ with probability $p$ (Definition~\ref{def:ping-failure}) and 
$0$ with probability $1-p$. Taking into account the fact that the indicator $m_{i}$ follows a Bernoulli trial with probability $I$ (see Appendix~\ref{appendix:proof-NIE-bias-rmse}), 
it follows that $\hat{m}_{i}$ follows also a Bernoulli trial with probability
\begin{align}\label{eq:prob-hat-m}
P\{\hat{m}_{i}\text{=1}\} &= P\{m_{i}\text{=1}\} + p \cdot P\{m_{i}\text{=0}\}
= I+(1-I) p
\end{align}
and thus it holds 
\begin{align}
E[\hat{m}_{i}] = P\{\hat{m}_{i}=1\} = I + (1-I)\cdot p \label{eq:expectation-hat-m}
\end{align}

In the case of measurement failures, the expression of NIE is calculated from the indicators $\hat{m}_{i}$. Thus, the expectation is
\begin{align}
E[\hat{I}] 
&= \textstyle E[\frac{1}{M}\sum_{i\in\mathcal{M}} \hat{m}_{i}] 
= \frac{1}{M}\cdot E[\sum_{i\in\mathcal{M}} \hat{m}_{i}] \nonumber\\
&\textstyle= \frac{1}{M}\cdot M\cdot \left(I+(1-I)\cdot p\right) = I+(1-I)\cdot p \label{eq:expectation-nie-failures}
\end{align}
where in the third equality we used the expression of \eq{eq:expectation-hat-m}.

\myitem{Bias.} The bias follows by substituting \eq{eq:expectation-nie-failures} in \eq{eq:bias-generic}:
\begin{align}
Bias_{NIE} = E[\hat{I}-I] = (1-I)\cdot p
\end{align}

\myitem{RMSE.} The RMSE of NIE, given the actual impact $I$, is:
\begin{align*}
RMSE(I)
&= \sqrt{E[(\hat{I} - I)^2]}= \sqrt{E[\hat{I}^2] + I^{2} - 2\cdot I\cdot E[\hat{I}]} 
\end{align*}
and using the property $Var(x) = E[x^2]-(E[x])^{2}$, gives:
\begin{align}
RMSE(I) &= \sqrt{Var[\hat{I}]+(E[\hat{I}])^{2} + I^{2} - 2\cdot I\cdot E[\hat{I}]} \nonumber\\
&\textstyle= \sqrt{Var[\hat{I}]+(E[\hat{I}]-I)^{2}}\nonumber\\
&\textstyle= \sqrt{Var[\frac{1}{M}\cdot\sum_{i\in\mathcal{M}} \hat{m}_{i}]+(I+(1-I)\cdot p-I)^{2}} \nonumber\\
&\textstyle= \sqrt{\frac{1}{M^{2}}\cdot Var[\sum_{i\in\mathcal{M}} \hat{m}_{i}]+((1-I)\cdot p)^{2}} \label{eq:rmse-I-failure-intermediate}
\end{align}
The quantity $\sum_{i\in\mathcal{M}} \hat{m}_{i}$ is the sum of $M$ independent Bernoulli trials with probability given by \eq{eq:prob-hat-m}. Therefore
\begin{align}
Var&\textstyle[\sum_{i\in\mathcal{M}} \hat{m}_{i}] 
= M\cdot P\{\hat{m}_{i}=1\}\cdot (1-P\{\hat{m}_{i}=1\})   \nonumber\\
&= M\cdot (I+(1-I)\cdot p)\cdot (1-I+(1-I)\cdot p) \label{eq:variance-sum-hat-m}
\end{align}
Substituting \eq{eq:variance-sum-hat-m} in \eq{eq:rmse-I-failure-intermediate}, 
and taking the expectation over the distribution $f(I)$, gives the expression of the theorem.



\section{The effect of the Public Monitors' Location on the NIE accuracy}\label{appendix:monitors-location}

\extratext{The simulation results in Fig.~\ref{fig:rmse-per-continent} validate that the lower accuracy of public monitors is due to their non uniform locations. We grouped hijacks based on the locations of the \{V,H\} ASes (i.e., continents where their headquarters are located\footnote{Retrieved from CAIDA's AS-Rank dataset \url{as-rank.caida.org}
}). We observe that when \{V,H\} reside in different continents (left subplot), the difference in the visibility from public monitors leads to higher RMSE (while the difference in random monitors is much smaller). The right subplot presents some indicative cases, where we see that when V or H are located in Asia (where public monitors are scarce) the RMSE is significantly higher than in the case where \{V,H\} are in N./S. America; on the contrary, random monitors yield similar accuracy in all cases.}

\begin{figure}[h]
\centering
\includegraphics[width=0.49\linewidth]{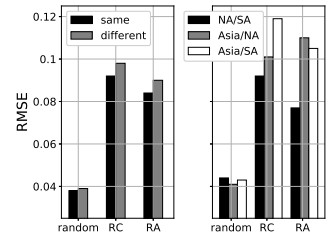}\label{fig:rmse-per-continent}
\caption{RMSE of NIE (y-axis) using $M=100$ monitors from different sets (x-axis); comparing results where \{V,H\} are located in the same/different continents (left subplot) and in N.America/S.America/Asia (right subplot).}
\label{fig:rmse-per-continent}
\end{figure}









\end{document}